\documentclass[12pt]{article}
\usepackage[margin=2cm]{geometry}
\usepackage{latexsym, amsfonts, amsthm,amssymb,amscd,amsmath,makeidx}
\usepackage{enumerate}
\usepackage{exscale}
\usepackage{color}  
\usepackage{graphicx} 
\usepackage{overpic} 
\usepackage{bm} 
\usepackage{bbm}
\usepackage{fancyhdr}
\usepackage{mathrsfs} 
\usepackage[normalem]{ulem}
\usepackage{exscale}
\usepackage{cleveref}

\definecolor{DarkGreen}{rgb}{0.2,0.6,0.2}

\def\eps{\varepsilon}
\def\Om{\Omega}

\numberwithin{equation}{section}

\def\Ind#1{{\mathbbmss 1}_{_{\scriptstyle #1}}}

\def\ignore#1{}

\def\bR{{\mathbb R}}

\def\cF{{\mathscr F}}

\def\cV{{\mathcal V}}

\def\<{\langle}\def\>{\rangle}

\newtheorem{theorem}{Theorem}[section]
\newtheorem{proposition}[theorem]{Proposition}

\theoremstyle{definition}

\newtheorem{remark}[theorem]{Remark}

 \normalsize

\usepackage{fancyhdr}
\title{Protecting Pegged Currency Markets\\ from Speculative Investors}
\author{ 
\normalsize Eyal Neuman\thanks{Department of Mathematics, Imperial College London, SW7 2AZ  London, United Kingdom}
 \and\setcounter{footnote}{6} \normalsize Alexander Schied\thanks{
		Department of Statistics and Actuarial Science, University of Waterloo. E-mail: {\tt alex.schied@gmail.com}. This author gratefully acknowledges financial support from the
 Natural Sciences and Engineering Research Council of Canada (NSERC) through grant RGPIN-2017-04054.
	}}
 
\date{\normalsize February 16, 2020}

\begin{document}

\maketitle
\vspace{-0.5cm}

\begin{abstract}
We consider a stochastic game between a trader and a central bank in a target zone market with a lower currency peg. This currency peg is maintained by the central bank through the generation of permanent price impact, thereby aggregating an ever increasing risky position in foreign reserves. We describe this situation mathematically by means of two coupled singular control problems, where the common singularity arises from a local time along a random curve. Our first result identifies a certain local time as that central bank strategy for which this risk position is minimized. We then consider the worst-case situation the central bank may face by identifying that strategy of the strategic investor that maximizes the expected inventory of the central bank under a cost criterion, thus establishing a Stackelberg equilibrium in our model.
\end{abstract}

\noindent{\it  Mathematics Subject Classification 2010:} 93E20,  91G80, 60J55 
\bigskip

\noindent{\it Key words:} currency peg, market impact, target zone models, singular stochastic control,  local time, Stackelberg equilibrium

\section{Introduction} \label{sec-intro} 
In September 16 of 1992, a day which is known as  Black Wednesday, the British government was forced to step out of the European Exchange Rate Mechanism (ERM). Shortly before this crucial date, the British pound exchange rate was close to its lower limit, and currency market speculators were trying to take advantage of this opportunity to drive down the exchange rate by making immense short orders. Among the other speculators stood out George Soros who shorted $10$ billion pounds within a single day. The British treasury kept spending its foreign currency during that day, in order to buy British pounds, which were becoming less valuable, as the exchange rate was rapidly decreasing. The Guardian  wrote  that \lq\lq 40 per cent of Britain's foreign exchange reserves were spent in frenetic trading" \cite{Guardian}. Also an announcement to raise interest rates to 15\% did not help. By the end of the day, the British government could not keep up with the purchases and announced  an exit from the ERM. As a result, the exchange rate dropped substantially (see Figure \ref{GBP-DEM Fig}) and the short positions of Soros and the speculators  surged in value.  In one day, George Soros made over $1$ billion Sterling. On the other hand, the cost for the British taxpayers was estimated at around 3.3 billion Sterling.

A similar situation occurred during the EUR/CHF currency peg, which was maintained by the Swiss National Bank (SNB) from September 6, 2011 to January 15, 2015. Here, however,  the roles between the foreign and domestic currencies were reversed, as euro investors scared by the European sovereign debt crisis were seeking a safe haven in the Swiss franc, thus driving up its value in comparison with the euro. In maintaining the currency peg, the SNB accumulated an ever increasing inventory, which in January 2015 stood at an amount equal to about 70\% of the Swiss GDP; see Figure \ref{EURCHF fig}. When the SNB finally abandoned the currency peg, the Swiss franc appreciated dramatically, and the FX inventory incurred a loss equal to about 12\% of the Swiss GDP \cite{LleoZiemba}.

The goal of this work is to describe such situations mathematically and to analyze the worst-case scenario a central bank may be facing when keeping up a currency peg. For simplicity, we focus on the situation in which the central bank tries to prevent an over-appreciation of  the domestic currency, as it was the case during the EUR/CHF currency peg. By switching signs and replacing long with short positions, one can easily see that our results also apply to the reversed situation as found, e.g., during the above-described exit of the British pound from the ERM.
We model the system of a speculative investor versus the central bank or government as a Stackelberg equilibrium within a stochastic differential game. One player,  the central bank, enforces the currency peg, thereby generating an ever increasing risky position in foreign currency. This position will incur dramatic losses once the currency peg is abandoned. The other player is a strategic investor whose strategy is aimed at increasing the central bank's risky position, in an attempt to force the central bank to terminate of the currency peg.

Our first result identifies a central bank strategy that, for all strategies of the strategic investor,  minimizes the accumulated FX inventory. Implementing this strategy will turn the exchange rate process into a reflecting diffusion process. As a byproduct, we thus give an endogenous justification for the common approach of modeling pegged currency markets as reflecting diffusions  (for such approaches, see,  e.g., \cite{Ball-Roma98,CarrKakushadze,NeumanSchied} and the references therein). Then we formulate an optimization problem for the strategic investor, who wishes to maximize the expected inventory of the central bank minus a trading cost term, by implementing a worst-case trading strategy. We give an explicit solution to this problem and thereby establish a Stackelberg equilibrium between our two players. 

Pegged exchange rates, also called \emph{target zone models}, have previously been studied in a wide range of contexts. For the corresponding economics literature, see \cite{Krugman91,Svensson, Bertola-Caballero92, De-Jong94, Ball-Roma98}, among others. Another stream of literature studies the central bank's problem in a setup of impulse control \cite{Jeanblanc-Picque1993,Korn1997}.  For two-sided target zones, i.e., exchange rates with both upper and lower currency pegs, one can use this approach to identify optimal peg levels \cite{Mundaca-Oksendal1998, Cadenillas-Zapatero1999, Cadenillas-Zapatero2000}. Yet, a third stream of literature studied optimal portfolio liquidation problems for investors who are active in an exogenously given target zone \cite{NeumanSchied,BelakMuhle-Karbe}. Finally, an optimal central bank strategy in a two-sided target zone is derived in \cite{NeumanSchiedWengXue}, and FX options pricing in target zones is studied in \cite{CarrKakushadze}.

  \begin{figure}
\begin{center}
\includegraphics[width=12cm]{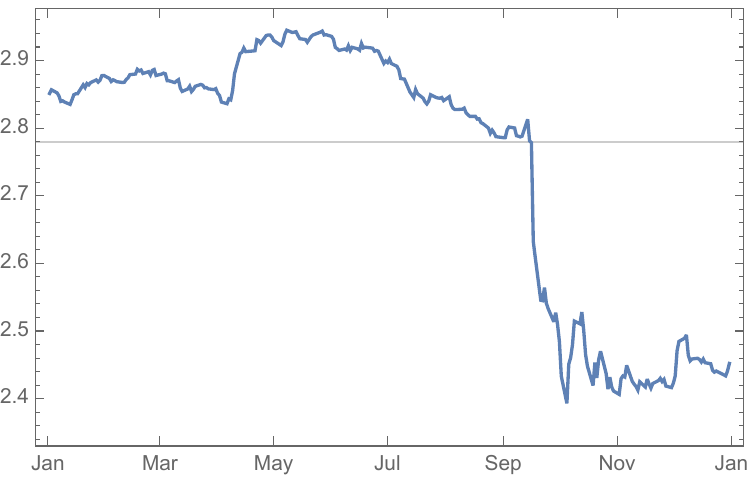}
\caption{Plot of the GBP/DEM exchange rate from  January 1, 1992  to December 31, 1992. On September 16, 1992, the British government  announced an exit from the European Exchange Rate Mechanism (ERM) in which the exchange rate was pegged at 2.78. A rapid drop of the exchange rate followed.}\label{GBP-DEM Fig}
\end{center}
\end{figure}

%
\begin{figure}
\begin{center}
\begin{overpic}[width=11cm]{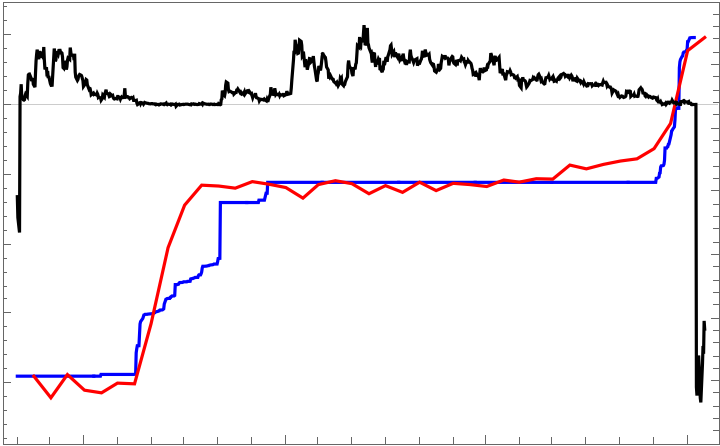}
\put(8,-3){\footnotesize 2012}
\put(36,-3){\footnotesize 2013}
\put(64,-3){\footnotesize 2014}
\put(92,-3){\footnotesize 2015}
\put(-7,8){\footnotesize 1.00}
\put(-7,17.5){\footnotesize 1.05}
\put(-7,27){\footnotesize 1.10}
\put(-7,37){\footnotesize 1.15}
\put(-7,46.5){\footnotesize 1.20}
\put(-7,56){\footnotesize 1.25}
\put(-18,61){\footnotesize EUR/CHF}
\put(101,8){\footnotesize 200}
\put(101,17){\footnotesize 250}
\put(101,25.5){\footnotesize 300}
\put(101,34.5){\footnotesize 350}
\put(101,43){\footnotesize 400}
\put(101,52){\footnotesize 450}
\put(101,63){\footnotesize FX reserves} 
\put(101,59){\footnotesize in bn EUR}
\end{overpic}
\end{center}
\caption{EUR/CHF exchange rate (black, left-hand scale)  during the time of the Swiss currency peg from September 2011 through January 2015. During this period, the exchange rate had a lower peg at 1.20.  A rapid increase of the exchange rate followed the announcement of the currency peg, whereas the rate dropped substantially after termination.
The red curve shows the Swiss FX reserves during that period. The blue curve shows the scaled empirical local time at $c=1.20$ of the exchange rate,  computed as in \eqref{emp local time eq}. It is affinely scaled so that it coincides with the Swiss FX reserves at the start and end of the  currency peg time period.
}\label{EURCHF fig}
\end{figure}

\section{Model setup and main results} \label{sec-res}

We consider the exchange rate between two currencies subject to a lower currency peg $c>0$. The exchange rate describes the price of one unit of foreign currency in terms of units of the domestic currency. The function of the lower currency peg is to keep the exchange rate above the threshold $c$ and thus to prevent on over-valuation of the domestic currency. Such currency pegs are frequently observed on financial markets.  Our reference example  is the lower EUR/CHF currency peg of 1.20 CHF, which was maintained by the Swiss National Bank (SNF) from September 6, 2011, to January 15, 2015; see Figure \ref{EURCHF fig}. By switching the roles taken by the foreign and domestic currencies, one can see that our approach will also apply to  currency pegs  where the undervaluation of the domestic currency is to be prevented. A corresponding example is s the GBP/DEM peg within the European Exchange Rate Mechanism. At least   \Cref{prop-skor} and  \Cref{local time thm} will also apply to two-sided pegs for which one of the two barriers is under attack, as these two results are formulated in a general setting that allows for an additional upper bound on the exchange rate.

We start with a very general model formulation and then will gradually add assumptions as our exposition proceeds. In the first step, we only need a single continuous trajectory $S:[0,\infty)\to[0,\infty)$, which models the exchange rate before any central bank intervention. We emphasize that, at this stage, no probabilistic or other assumptions on the dynamics of $S$ are needed. 

If the exchange rate comes too close to the threshold $c$, the central bank will have to intervene so as to defend the currency peg. The two main instruments at the disposal of the central bank are:
\begin{enumerate}
\item lowering the domestic reference interest rate;
\item generating price impact by selling the domestic currency and buying the foreign currency. 
\end{enumerate}
 Method (a) is difficult to implement and may have substantial side effects. As a matter of fact, in our reference example of the EUR/CHF currency peg, the SNB  used this instrument only once during the duration of the currency peg: On December  18, 2014, i.e., less than one month before the abandonment of the currency peg, the SNB decreased its 3-months Libor target range from  $[0.0\%,0.25\%]$ to  $[-0.75\%,-0.25\%]$.  
Although this decision made headlines by introducing a regime of strictly negative interest rates,  one can see from Figure \ref{EURCHF fig} that  it had no visible effect on the exchange rate, which increased from 1.2010 on December 17 to a mere 1.2039 on December 18. Also,  during the \lq\lq breaking of the bank of England", announcing a dramatic increase of interest rates to the level of 15\% could not keep the British pound from dropping further against the Deutsche mark \cite{Guardian}.  For this reason, we will ignore interest rates in our setup and focus exclusively on central bank intervention through the generation of permanent price impact.  

\subsection{An optimal central bank strategy}

In this section, we focus on  permanent price impact that is linear in the trade size. That is, after buying $y$ units of foreign currency, the exchange rate will be permanently increased by $\gamma y$, where $\gamma>0$ is a certain constant. The linearity of permanent price impact is a standard assumption in price impact models; for a discussion, we refer to \cite[Section 3]{GatheralSchiedSurvey} and the references therein. To keep the exchange rate above the threshold $c$, the central bank will thus seek a trading strategy whose inventory $Y_t$ in foreign currency at time $t$ is such that 
\begin{equation} \label{cond-c} 
S^Y_t:=S_t+\gamma Y_t\ge c,\qquad \text{for all $t\ge0$.}
\end{equation}
In addition, we assume that $Y_0=0$ and that $ Y_t$ is a continuous and nondecreasing function of $t$. At first glance,  the assumption of monotonicity may not be intuitive, because it might be beneficial for the central bank to unload problematic inventory when the exchange rate drifts away from the threshold $c$.  However, such a move could be difficult politically, as unloading inventory will in turn lower the exchange rate and thus  might   be perceived as acting against the initial objective or even against national interest. The latter argument will apply  in particular, but not exclusively, to an export-driven economy such as Switzerland's. As a matter of fact, the monotonicity assumption was basically satisfied during the EUR/CHF currency peg, as can been seen from Figure \ref{EURCHF fig}, where the red line represents the development of the Swiss FX reserves\footnote{The small fluctuations that can been seen in that plot may even stem from other sources than from an unloading of inventory by the SNB. They may be due to value fluctuations of reserves held in other currencies than EUR, SNB trading activity,  or from fluctuations due to international trade.}. 

 The inventory accumulated through a central bank  strategy $Y$ can be problematic. This was obviously the case during \lq\lq the breaking of the Bank of England" by the investor George Soros, albeit of course with an inverted sign: According to \cite{Guardian}, it was estimated that \lq\lq 40 per cent of Britain's foreign exchange reserves were spent in frenetic trading". For our EUR/CHF reference example, we can observe from Figure \ref{EURCHF fig} that the abandonment of the Swiss currency peg was preceded by another surge in FX inventory by about 100 billion euros. The total FX inventory stood then at 480 billion euros, an amount about equal to 70\% of the Swiss GDP.\footnote{See also the article  \emph{Why the Swiss unpegged the franc}, published on January 18, 2015 in 
The Economist.} Following the abandonment of the currency peg on January 15, 2015, the FX inventory of the Swiss National bank (SNB) lost around CHF 78 billion of its value, which amounts to about 12\% of the annual Swiss GDP; see  \cite{LleoZiemba}. 
Based on these observations, it is natural for the central bank to adopt a strategy that minimizes the accumulated inventory over all possible strategies.  The following proposition states that this goal can be achieved for an arbitrary exchange rate process by using the  Skorokhod map.

\begin{proposition}\label{prop-skor} 
Assume that $ S_0> c$. Then there exists  a unique  continuous and nondecreasing strategy $Y^*$ that satisfies \eqref{cond-c} 
 and that minimizes the central bank inventory  in the sense that $ Y^{*}_t\le Y_t$ for all $t\ge0$ and for all other continuous and nondecreasing strategies $Y$ satisfying \eqref{cond-c}. Moreover, $Y^*$ is given by \begin{equation} \label{opt-y}
 Y^{*}_{t} = \frac{1}{\gamma}\big(  \max_{0\leq r\leq t} \{c-  S_r\}\big)_{+}, \qquad t\ge0,
\end{equation}
where  $x_+=\max\{0,x\}$ for $x\in\bR$.
\end{proposition}

\begin{remark} \Cref{prop-skor} has also the following  significance: Exogenously given pegged exchange rate processes are often modeled by means of reflected diffusion processes; see, e.g., \cite{Ball-Roma98,CarrKakushadze,NeumanSchied} and the references therein. Our \Cref{prop-skor}  now shows that exactly such a model appears endogenously as solution to an optimal central bank strategy, because reflecting diffusion processes arise by applying the Skorokhod map to a certain diffusion process (see, e.g., \cite[p.~5]{Pilipenko}).
\end{remark}

\bigskip

\Cref{prop-skor}  has one drawback: The optimal response $Y^{*}$ in \eqref{opt-y} is a priori a functional of $S$. The trajectory $S$, however,  is not observable, because it describes the exchange rate fluctuation \emph{before} central bank observation. The following result shows that, in the context of a semimartingale model, the optimal response $Y^{*}_t$ can be recovered from the actual exchange rate, $S^{Y^*}$,  \emph{post} central bank intervention. We therefore now fix a filtered probability space $(\Omega,\cF,(\cF_t)_{t\ge0},P)$ for which $\cF_0$ is $P$-trivial and assume that $S$ is a continuous semimartingale on that probability space. Recall  that  $S$ satisfies the \emph{structure condition}, if its semimartingale decomposition $S_t=M_t+A_t$ into a continuous local martingale $M$ and a continuous adapted process $A$ with sample paths of locally bounded variation is such that $dA_t\ll d\<S\>_t=d\<M\>_t$ $P$-a.s. Recall also  that the local time in $c$ of any continuous semimartingale $X$ can be obtained as
\begin{equation}\label{local time def}
L^c_t(X)=\lim_{\eps\downarrow0}\frac1\eps\int_0^t\Ind{[c,c+\eps)}(  X_r)\,d\<  X\>_r.
\end{equation}
 The mathematical content of the following result  extends an analogous result  for solutions of reflecting SDEs; see, e.g., Theorem 1.3.1 in \cite{Pilipenko}. In light of \eqref{local time def}, formula \eqref{Y feedback eq} in the following proposition can be regarded as  a representation of the optimal central bank strategy $Y^*$ in singular feedback form.

\begin{theorem}\label{local time thm}If  the semimartingale $S$ satisfies the structure condition, then the optimal central bank strategy satisfies 
\begin{equation}\label{Y feedback eq}
Y^*_t=\frac1{2\gamma}L^c_t(S^{Y^*}),\qquad \text{for all $t\ge0$ $P$-a.s.}
\end{equation}
\end{theorem}

\bigskip

By \eqref{local time def}, the empirical local time of the daily EUR/CHF closing prices $(S_{t})$ should be approximately proportional to 
\begin{equation}\label{emp local time eq}
\widehat L^c_t=\frac1{0.004}\sum_{r\le t}\Ind{[c,c+0.004]}(S_r)(S_r-S_{r-1})^2,
\end{equation}
where $c=1.20$. Here the proportionality factor will depend on $\eps = 0.004$ but, as suggested by the results of \cite{bass},  also on time discretization. In \Cref{EURCHF fig}, we have used this formula to plot the empirical local time along the EUR/CHF  exchange rate and the Swiss FX reserves after applying an affine scaling  so that the scaled local time coincides with the Swiss FX reserves at the start and end of the  currency peg time period.  More precisely, the Swiss FX inventory in October 2011 was 204.848 billion EUR, and in February 2015 it had increased to 472.684 billion EUR. Thus, the blue curve in Figure \ref{EURCHF fig} is 
$$204.848+\frac{472.684-204.848}{\widehat L^c_T}\cdot \widehat L^c_t,
$$
where $T=\,$January 14, 2015. As one can see from Figure \ref{EURCHF fig}, the fit between the corresponding curve (blue in Figure \ref{EURCHF fig}) and the true FX inventory is rather striking, despite the fact that our data grids are rather coarse (daily exchange rates and monthly data for the FX inventory).
We can therefore conclude that 
$$Y_T=(472.684-204.848)\, \text{bn EUR}
=267.836\, \text{bn EUR}$$
is a good approximation to the total SNB  response during the EUR/CHF peg. 

We finally remark that the \lq local time strategy\rq\ for the central bank can be implemented in a straightforward manner. All the central bank needs to do is placing and maintaining a very large limit order at the level $c$. Such a passive strategy does not require to cross the spread, which justifies in hindsight our implicit assumption that there are no transaction costs for the central bank strategy.

\subsection{A Stackelberg equilibrium between central bank and strategic investor}

In this section, we add another player to our setup. This player could be either a large strategic investor, a group of such investors, or simply for a strong macroeconomic trend pushing the exchange rate toward the threshold $c$. The role of George Soros and other opportunistic traders during the breaking of the Bank of England in September 1992 could be an example for the type of strategic investor we have in mind, while the flight into Swiss francs during the European Sovereign Debt Crisis from 2011 onward created a strong macroeconomic trend leading to a decrease of the EUR/CHF rate and the resulting inception of the currency peg. In the sequel, we will always refer to a \lq strategic investor\rq, regardless of whether the corresponding actions are driven by a single fund or by macroeconomic circumstances. 

To model this situation, we assume henceforth that the filtered probability space $(\Om,\cF,(\cF_t)_{t\ge0},P)$ satisfies the usual conditions  and that the exchange rate process before central bank intervention is given by a Bachelier model with  drift $\gamma\int_0^tv_r\,dr$, which models the cumulative permanent price impact generated by the strategic agent's trading strategy in which $v_t\,dt$ of the foreign currency are bought at time $t$ (if $v_t$ is negative, this transaction will correspond to a sale). We assume that  $v$  is progressively measurable, and satisfies $ 
E[\int_{0}^{t}v^{2}_{s}\,ds ] <\infty$ for all $t\ge0$. 
The exchange rate dynamics before central bank intervention are thus given by
\begin{equation}\label{Sv eq}
\bar S^v_t=S_0+\sigma W_t+\gamma\int_0^tv_s\,ds,
\end{equation}
where $W$ is a standard $(\cF_t)_{t\ge0}$-Brownian motion  and $\sigma$ is a positive constant. 
The central bank, in turn, will now react to the  dynamics \eqref{Sv eq} by implementing the optimal strategy from \Cref{prop-skor}  and \Cref{local time thm}, which we will denote by $Y^v$, so as to make its dependence on $v$ explicit. The actual exchange rate process will thus be given by 
\begin{equation} \label{s-bm}
S^{Y^v}_t=\bar S^v_t+\gamma Y^v_t=S_0+\sigma W_t+\gamma\int_0^tv_s\,ds+\gamma Y^v_t,\qquad t\ge0.
\end{equation}

\begin{remark} Let us comment on some aspects of our modeling framework. First, our description of linear permanent price impact is drawn from Bertsimas and Lo \cite{BertsimasLo} and the  linear, continuous-time Almgren--Chriss model \cite{Almgren}. Second, there are several arguments in favor of the Bachelier model and against  geometric Brownian motion (GBM) in foreign exchange modeling. For instance, in the long run, the sample paths of  GBM either grow exponentially or decline to zero. While this behavior is consistent with many equity price trajectories, it is usually not observed for exchange rates. Moreover, the absolute fluctuations of a stock price process are typically proportional to the current price. This makes sense economically, because then the fluctuations of a certain amount of capital invested into the stock will depend only on the amount of capital and not on the unit price of the stock. GBM achieves this effect. For exchange rates, however, this intuition fails, and  an empirical analysis of common currency pairs reveals that the sizes of their fluctuations remain more or less independent of the exchange rate levels. Moreover, in our special situation, one may expect that the exchange rate process $S^{Y^v}$ will stay in close proximity of the barrier $c$, and so the assumption of a constant volatility in \eqref{s-bm} is natural from a modeling perspective and  not an oversimplification. As a matter of fact, Figure~\ref{EURCHF fig}
shows that the EUR/CHF exchange rate stayed within the narrow window $[1.20,1.26]$ throughout the duration of the Swiss currency peg.

\end{remark}

\medskip

Now we turn toward the optimization problem of the strategic investor.  Our goal is to identify the worst-case scenario with which the central bank has to reckon when keeping up a (one-sided) currency peg. To this end, the goal of the strategic investor will consist in maximizing the future expected inventory $E[Y^{v}_T]$ of the central bank in an attempt to force the central bank to abandon the target zone. Here, $T$ is a certain time horizon picked by the strategic investor or by the central bank. The strategic investor may sometimes have to cross the spread or even trade deeper into the order book so as to drive  the exchange rate down toward the barrier $c$. In contrast to the central bank's strategy, the investor's strategy $v$ will thus create transaction costs, which are sometimes also called  \lq\lq slippage". According to the linear Almgren--Chriss model \cite{Almgren}, this  slippage is proportional to
$
\int_0^Tv_t^2\,dt.
    $
We therefore assume that, for a certain constant $\kappa>0$, the goal of the investor is to 
\begin{equation}\label{rev-spec}
\text{maximize}\qquad E\bigg[ \gamma Y^{v}_{T}-\kappa\int_0^Tv_t^2\,dt\bigg]\qquad\text{over $v\in\cV_T $,}
\end{equation} 
where $\cV_T$ denotes the class of progressive strategies $v$ with $
E[\int_{0}^{T}v^{2}_{s}\,ds ] <\infty$ 
Our main result provides a solution to the preceding problem of stochastic optimal control in terms of a solution to a Skorokhod SDE, and thus establishes a Stackelberg equilibrium in our singular stochastic differential game between a  trader and the central bank. 

\bigskip

\begin{theorem} \label{theorem-trader} Suppose that  $S_0=z> c$ and let $\beta=\gamma^2/(2\kappa\sigma^2)$ and
\begin{equation} 
\label{V-00}
U(t,y):= \frac{1}{\beta}\log \Big(E\Big[\exp\big( \sigma \beta  L^{(y-c)/\sigma}_t( W)\big)\Big] \Big),\qquad y\in\mathbb R,\ t\ge 0,
\end{equation}
where $L^x_t(W)$ is the local time of the standard Brownian motion $W$ at level $x\in \mathbb{R}$.
Then we have 
$$
U(t,z)=\sup_{v\in\cV_t }E\Big[\,\gamma Y^{v}_{t}-\kappa\int_0^tv_s^2\,ds\,\Big]\qquad \text{for $t\ge0$.}
$$
Moreover, $U(t,z)$ is smooth on $[0,\infty)\times[c,\infty)\setminus\{(0,c)\}$ and 
when letting
\begin{equation}\label{dif inclusion}
\bar v(t,z):=\frac\gamma{2\kappa}\partial_zU (t,z),
\end{equation}
then, for each time horizon $T>0$, there exists a unique strong solution to the following  SDE with reflection,
\begin{equation}  \label{skor-sde} 
\begin{cases}
dS_{t}&=\sigma \,dW_t+\gamma \bar v(T-t,S_t)\,dt +dR_{t},\\
S_t&\ge c,\\
S_0&=z,\\
R_t&\text{is continuous, nondecreasing, and }\textstyle\int_0^T\Ind{\{S_t>c\}}\,dR_t=0.
\end{cases}
\end{equation} 
Then 
$v^*:=\bar v(T-t,S_t)$
belongs to $\cV_T$ and is an optimal strategy for the strategic investor with time horizon $T$.  Finally, the corresponding optimal central bank response \eqref{opt-y} is given by $Y^{v^*}=\frac1\gamma R$.
\end{theorem}
\begin{remark} From Formula 1.3.3 on p.~161 in~\cite{BorodinSalminen} we get a closed-form expression for $U(t,z)$ if $t>0$, 
\begin{align}\label{U explicit formula eq}
U(t,z)=\frac1\beta\log\Bigg(\mathrm{erf}\Big(\frac{z-c}{\sigma\sqrt{2t}}\Big)+e^{-\beta(z-c)+\beta^2\sigma^2t/2}\bigg[1-\mathrm{erf}\Big(\frac{z-\beta\sigma^2 t}{\sigma\sqrt{2t}}\Big)\bigg]\Bigg),
\end{align}
where $\mathrm{erf}(x)=\frac2{\sqrt\pi}\int_0^xe^{-y^2}\,dy$ is the Gaussian error function. It follows that
\begin{equation}
\bar v(t,z)=-\frac{\gamma e^{\frac{1}{2} \beta ^2 \sigma ^2 t} \left(\text{erf}\left(\frac{c+\beta  \sigma
   ^2 t-z}{\sqrt{2} \sigma  \sqrt{t}}\right)+1\right)}{2\kappa e^{\frac{1}{2} \beta ^2 \sigma ^2
   t} \left(\text{erf}\left(\frac{c+\beta  \sigma ^2 t-z}{\sqrt{2} \sigma 
   \sqrt{t}}\right)+1\right)+e^{\beta  (z-c)} \text{erf}\left(\frac{z-c}{\sqrt{2} \sigma 
   \sqrt{t}}\right)}.
\end{equation}
See Figure~\ref{Uv figure} for  plots of $U(t,z)$ and $\bar v(t,z)$. 
\end{remark}

\begin{figure}
\centering
\includegraphics[width=8cm]{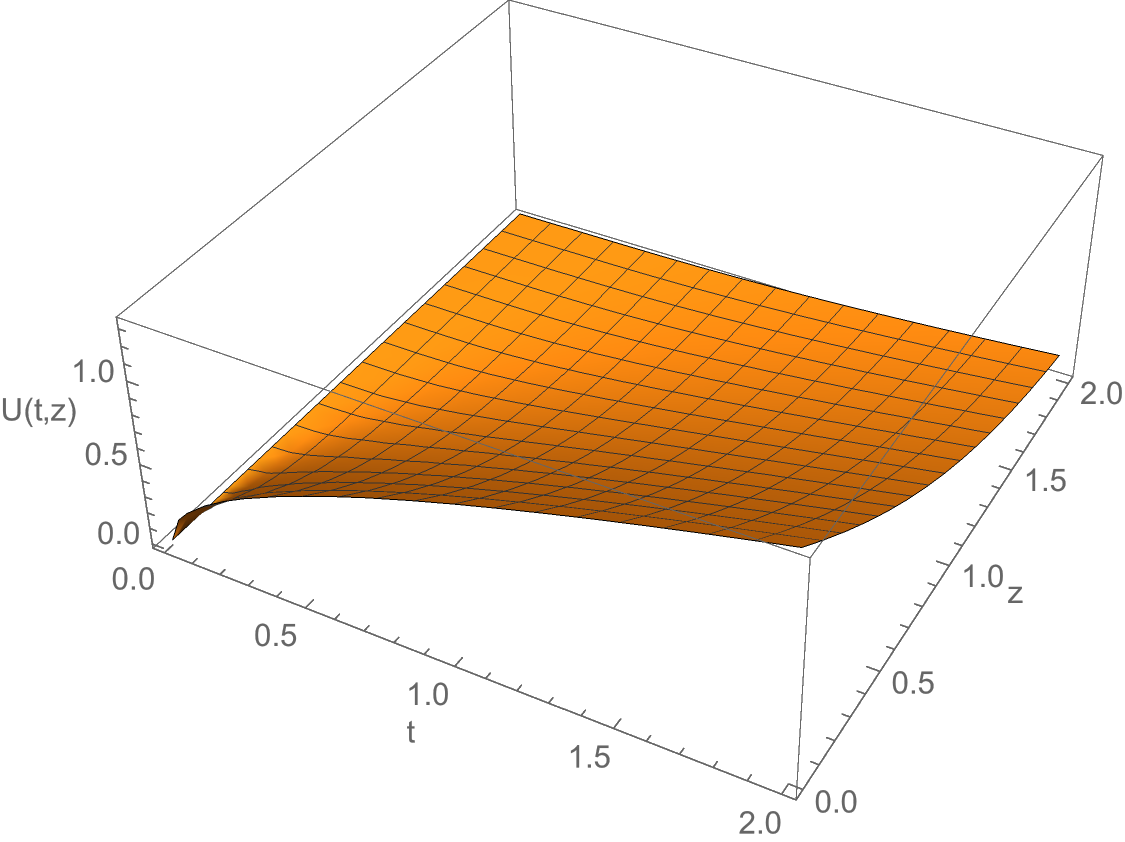}\qquad
\includegraphics[width=8cm]{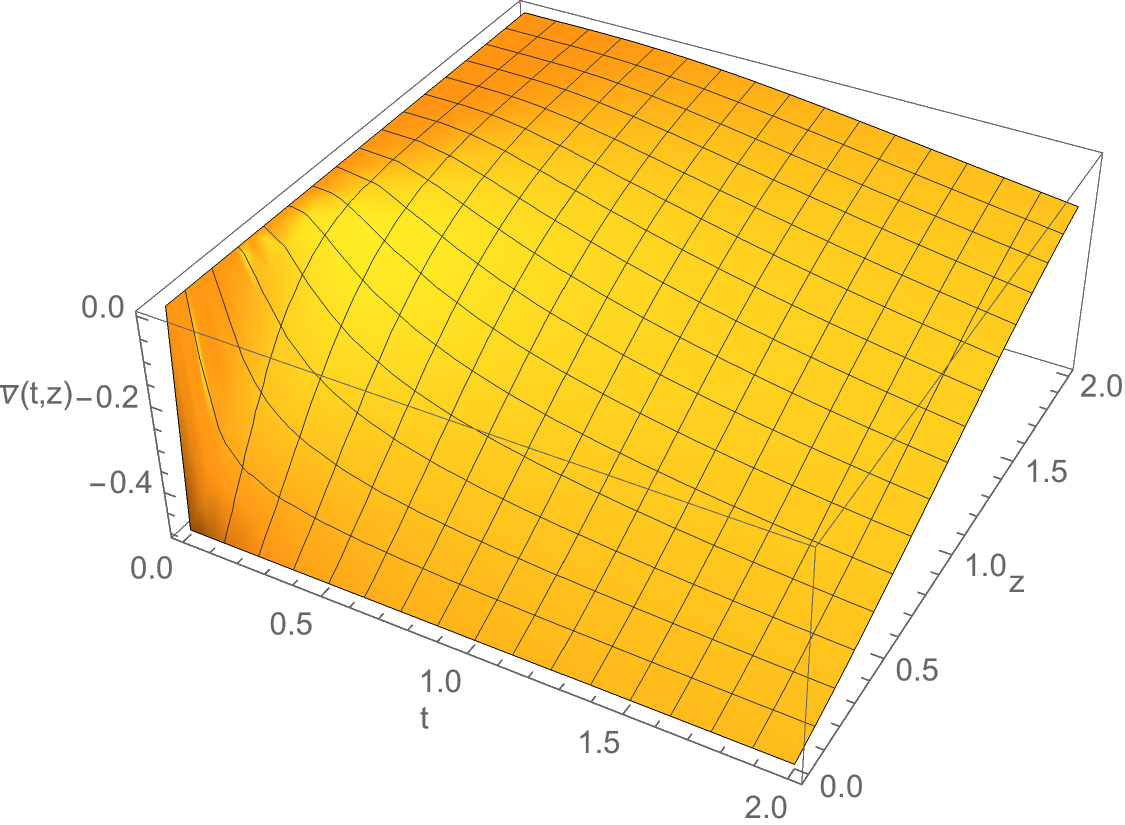}
\caption{The value function $U(t,z)$ (left) and the function $\bar v(t,z)$ (right) for $\sigma=\gamma=\kappa=1$ and $c=0$}. 
\label{Uv figure}
\end{figure}

\begin{figure}
\centering
\includegraphics[width=12cm]{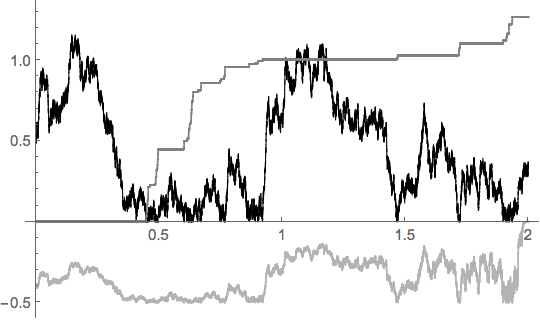}
\caption{Solution  of the reflecting SDE \eqref{skor-sde} (black), optimal central bank response, $R$ (dark gray, nondecreasing), and optimal investor strategy, $v^*$ (light gray, lower trajectory), for $\sigma=\gamma=\kappa=1$, $c=0$, and $T=2$.  }. 
\label{trajectories figure}
\end{figure}

\begin{figure}
\centering
\begin{minipage}[b]{8cm}
\includegraphics[width=8cm]{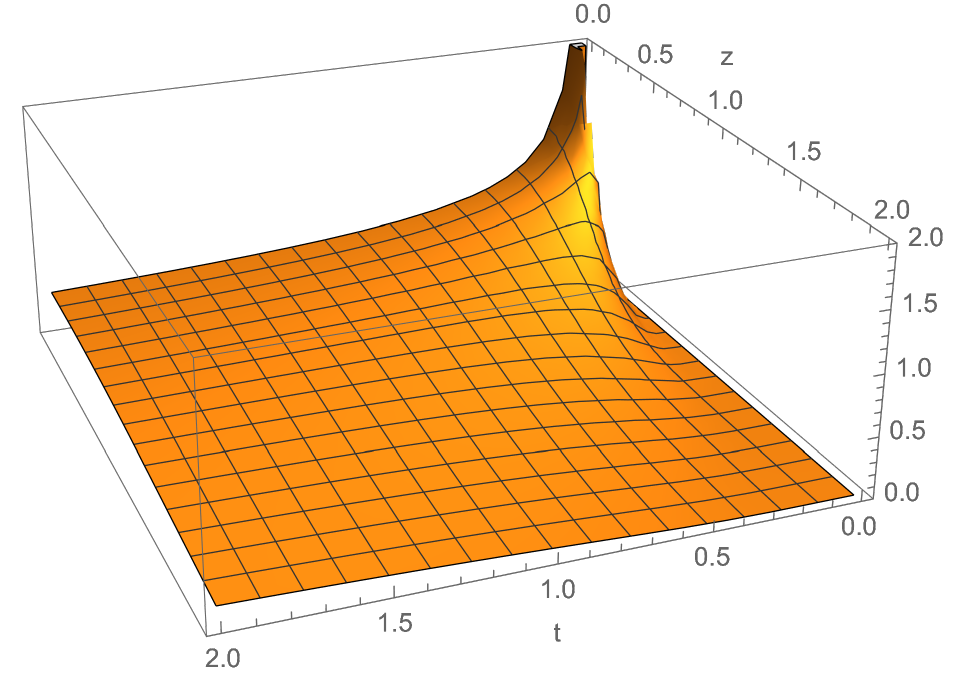}
\end{minipage}\quad 
\begin{minipage}[b]{8cm}
\includegraphics[width=8cm]{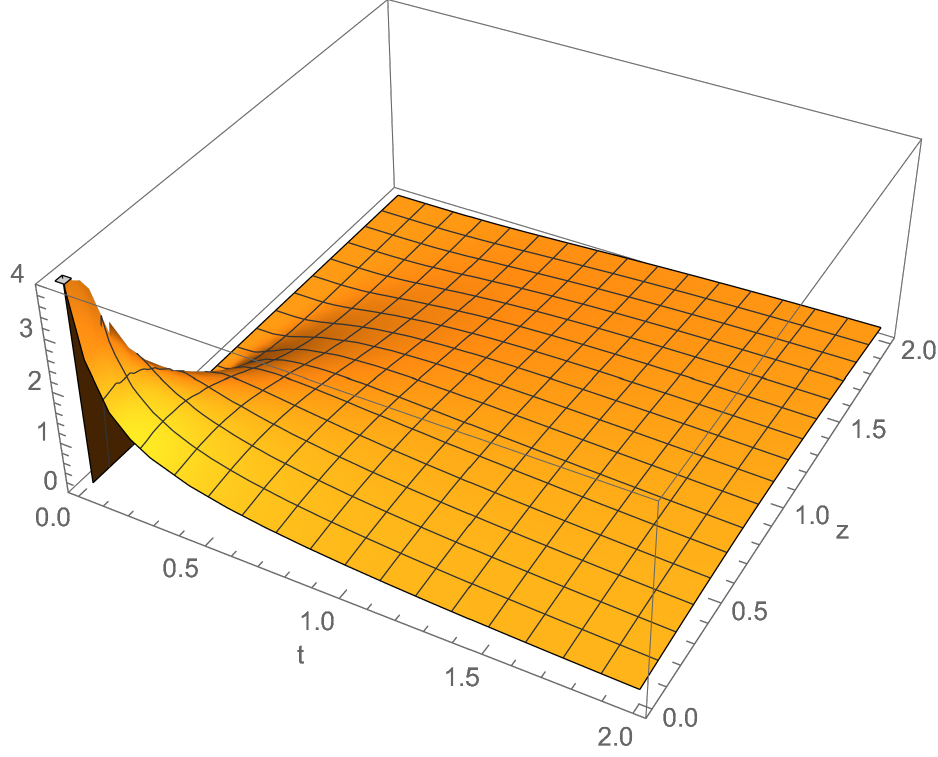}
\end{minipage}
\\
\begin{minipage}[b]{8cm}
\includegraphics[width=8cm]{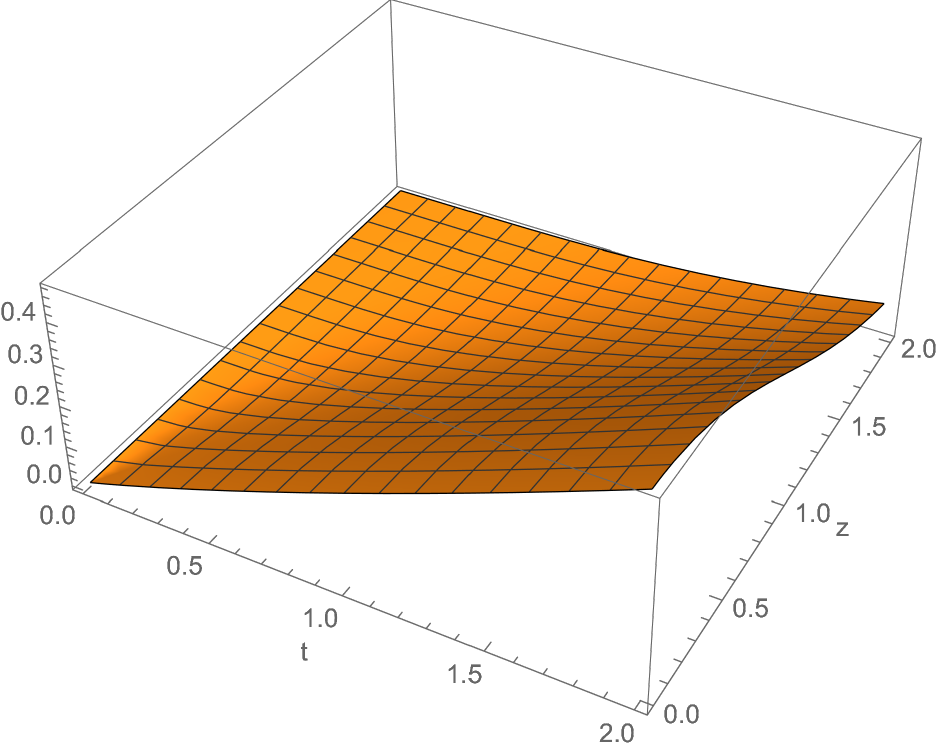}
\end{minipage}
\quad 
\begin{minipage}[b]{8cm}
\includegraphics[width=8cm]{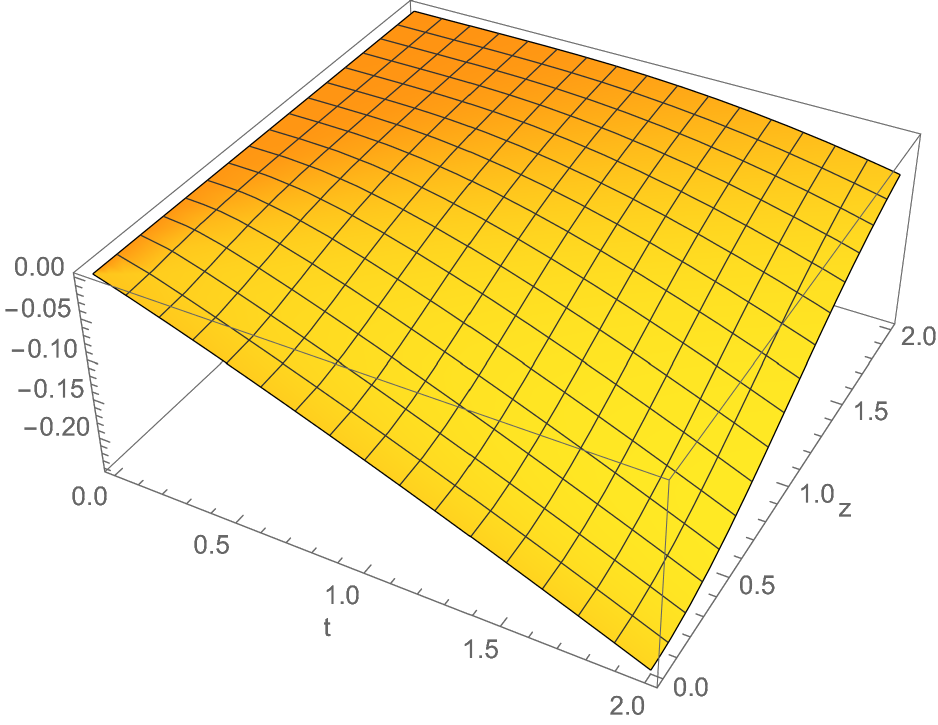}
\end{minipage}\\
\begin{minipage}[b]{8cm}
\includegraphics[width=8cm]{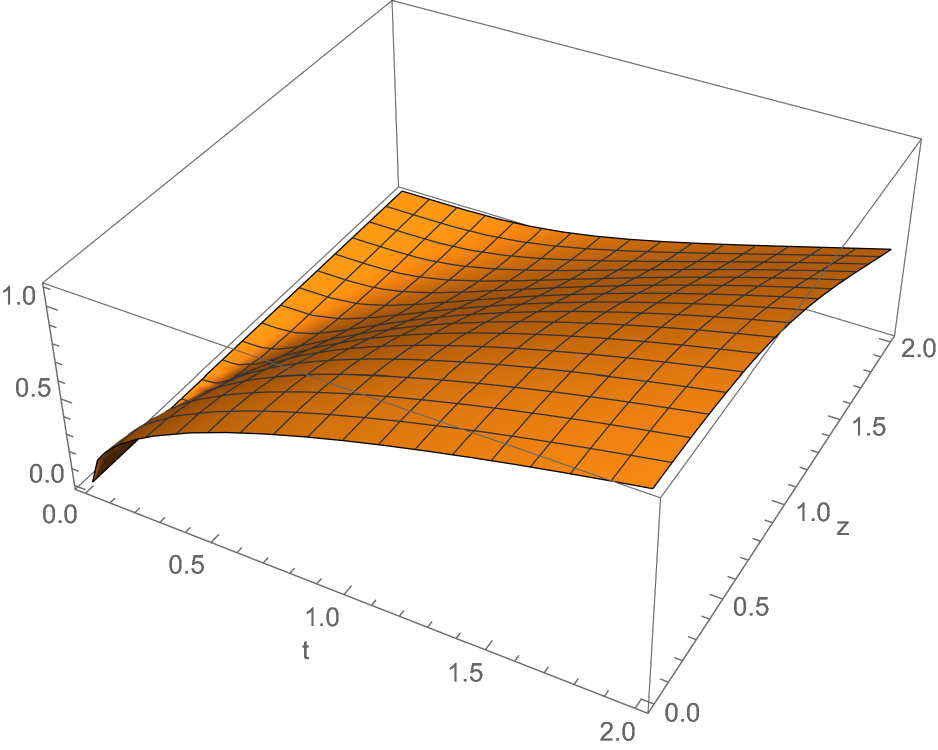}
\end{minipage}
\quad
\begin{minipage}[b]{8cm}
\includegraphics[width=8cm]{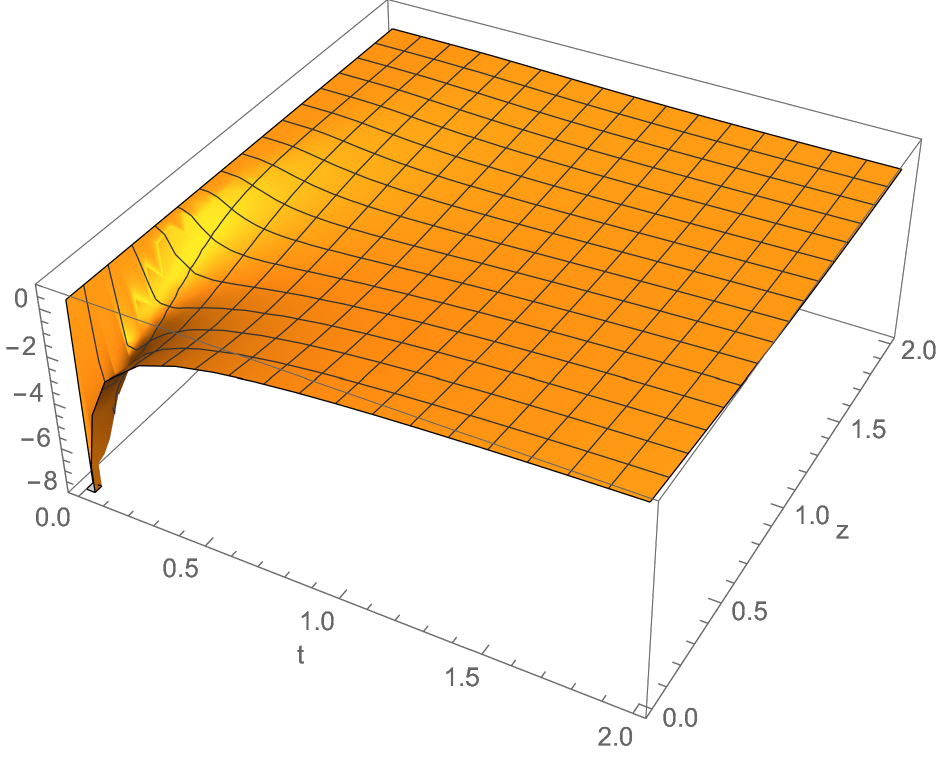}
\end{minipage}
\caption{Derivatives of the value function for $\sigma=\gamma=\kappa=1$ and $c=0$ from top left to bottom right: $\partial U/\partial t$, $\partial ^2U/\partial z^2 $, $\partial U/\partial \gamma$, $\partial U/\partial \kappa$, $\partial U/\partial \sigma$, and $\partial U/\partial c$. }
\end{figure}

\section{Proofs}

\begin{proof}[Proof of  \Cref{prop-skor}] 
Note that 
\begin{equation} \label{g-proc}
S^{Y^*}_t= S_t +\gamma Y_{t}^{*} \end{equation}
is the Skorokhod map of the continuous function $S$. Hence $S^{Y^*}_t\ge c$ for all $t$, and 
the minimality of $Y^{*}$  follows from the minimality property of the solution $(S^{Y^*},\gamma Y^{*})$ to the Skorokhod problem as stated, e.g., in Remark 6.15 in \cite{KaratzasShreve}.
 \end{proof}

\begin{proof}[Proof of \Cref{local time thm}.] 
By Tanaka's formula,  the local time of $ S^*:=S^{Y^*}$ satisfies
\begin{equation*}
\frac12 L^c_t(S^*)=( S^*_t-c)^+-(  S^*_0-c)^+-\int_0^t\Ind{\{  S^*_r>c\}}\,dS^*_r.
\end{equation*}
We have $S^*:= S+\gamma Y^{*}\ge c$  $P$-a.s. Next, the well-known properties of the Skorokhod map imply that the measure $dY^*_t$ is $P$-a.s.~concentrated on the set $\{r\,|\,S^*_r=c\}$. Hence,
\begin{align*}
\frac12 L^c_t(S^*)&=  S^*_t- S^*_0-\int_0^t\Ind{\{  S^*_r>c\}}\,dS^*_r\\
&=S^*_t- S^*_0 -\int_0^t\,dS^*_r +\int_0^t\Ind{\{  S^*_r=c\}}\,dS^*_r\\
&=\int_0^t\Ind{\{  S^*_r=c\}}\,dS^*_r\\
&=\gamma \int_0^t\Ind{\{  S^*_r=c\}}\,dY^*_r+\int_0^t\Ind{\{  S^*_r=c\}}\,d S_r\\&=\gamma Y^*_t+\int_0^t\Ind{\{  S^*_r=c\}}\,dM_r+\int_0^t\Ind{\{  S^*_r=c\}}\,dA_r.
\end{align*}
Here, $ S=M+A$ denotes the semimartingale decomposition of $ S$. The semimartingale decomposition of $S^*$ is 
$$S^*=M+(A+\gamma Y^*),$$
 and so $\<S^*\>=\< S\>=\<M\>$. Therefore, the occupation time formula for semimartingale local times gives that $P$-a.s.
$$0=\int_{-\infty}^\infty \Ind{\{  x=c\}} L^x_t(S^*)\,dx=\int_0^t\Ind{\{  S^{*}_r=c\}}\,d\<S^*\>_r=\int_0^t\big(\Ind{\{  S^{*}_r=c\}}\big)^2\,d\<M\>_r.
$$
The It\^o isometry hence implies that 
$\int_0^t\Ind{\{  S^*_r=c\}}\,dM_r=0$ $P$-a.s. Moreover, our assumption that $ S$ satisfies the structure condition yields  that $\int_0^t\Ind{\{  S^{*}_r=c\}}\,dA_r=0$. Thus, we must have $L^c_t(S^*)=2\gamma Y^*_t$.\end{proof}

\subsection{ Proof of Theorem \ref{theorem-trader}} 

\begin{proposition} \label{prop-sing-pde}
The function $U$ satisfies the following partial differential equation,
\begin{equation}\label{measure HJB}
\partial_t U(t,z)=\frac{\sigma^2}2\partial_{zz}U(t,z)+\frac{\gamma^2}{4\kappa}\big(\partial_zU^{}(t,z)\big)^2,\qquad\text{in $(0,\infty)\times[c,\infty)$,}
\end{equation}
with boundary and initial conditions 
\begin{equation}\label{measure HJB boundary condition}
\partial_z U(t,c)=-1,\quad t\in(0,T], \qquad U(0,z) =0, \quad z \geq c.
\end{equation}
\end{proposition}

\begin{proof}
Let 
$$\psi(t,x):=\mathrm{erf}\Big(\frac{x}{\sigma\sqrt{2t}}\Big)+e^{-\beta x+\beta^2\sigma^2t/2}\bigg[1-\mathrm{erf}\Big(\frac{x-\beta\sigma^2 t}{\sigma\sqrt{2t}}\Big)\bigg].
$$
For $t>0$, calculations show that
\begin{align*}
\partial_t\psi(t,x)&=\frac{\beta  \sigma e^{-\frac{x^2}{2 \sigma ^2
   t}}}{\sqrt{2\pi t}} +\frac{1}{2} \beta^2  \sigma^2  e^{-\beta x+\beta^2\sigma^2t/2}\bigg[1-\mathrm{erf}\Big(\frac{x-\beta\sigma^2 t}{\sigma\sqrt{2t}}\Big)\bigg],\\
\partial_x\psi(t,x)&=-\beta e^{-\beta x+\beta^2\sigma^2 t/2}\bigg[1-\mathrm{erf}\Big(\frac{x-\beta\sigma^2 t}{\sigma\sqrt{2t}}\Big)\bigg],\\
\partial_{xx}\psi(t,x)&=\frac2{\sigma^2}\partial_t\psi(t,x).
\end{align*}
From (\ref{U explicit formula eq}) we have $U(t,z)=\frac1{\beta}\log\psi(t,z-c)$. It follows that 
\begin{align}\label{first U derivatives}
\partial_tU(t,z)=\frac{\partial_t\psi(t,z-c)}{\beta \psi(t,z-c)},\quad
\partial_z U(t,z)=\frac{\partial_x\psi(t,z-c)}{\beta \psi(t,z-c)}.
\end{align}
In particular, 
$$\partial_z U(t,c)=\frac{\partial_x\psi(t,0)}{\beta \psi(t,0)}=-1.
$$
Next, the second $z$-derivative of $U$ corresponds to
\begin{equation*}
\begin{split}
\partial_{zz}U(t,z)&=\frac{\partial_{xx}\psi(t,z-c)}{\beta \psi(t,z-c)}-\beta\bigg(\frac{\partial_x\psi(t,z-c)}{ \beta\psi(t,z-c)}\bigg)^2.
\end{split}
\end{equation*}
Plugging everything together yields the assertion.
\end{proof}

Now we are ready to prove Theorem~\ref{theorem-trader}. 

\paragraph{Proof of Theorem~\ref{theorem-trader}}
Recall that $S_0=z>c$ and define $V(t,z)$ to be the trader's value function from (\ref{rev-spec}), that is 
$$
V(t,z):=\sup_{v\in\cV_t }E\Big[\, \gamma Y^{v}_{t}-\kappa\int_0^tv_s^2\,ds\,\Big].
$$
In a first step, we show that $V(t,z)$ is equal to the right-hand side of \eqref{V-00}. To this end, recall that $\beta=\gamma^2/(2\kappa\sigma^2)$ and use (\ref{s-bm}) and (\ref{opt-y}) to get 
$$
\begin{aligned} 
-\beta V(t,z)&=\inf_{v\in\cV_t }E_{}\Big[\, - \beta  \gamma Y^{v}_{t}+\beta\kappa \int_0^tv_s^2\,ds\,\Big] \\
&=\inf_{v\in\cV_t }E_{}\Big[\,- \beta \big(  \max_{0\leq s\leq t} \{c- \bar S^{v}_t\}\big)_{+}+\frac{\gamma^2}{2\sigma^2}\int_0^tv_s^2\,ds\,\Big] \\ 
&= \inf_{v\in\cV_t }E_{}\Big[\,-\beta \sigma \Big(  \max_{0\leq s\leq t}\Big \{\frac{c- z}{\sigma} - W_{s}-  \int_{0}^{s}\big(\frac{\gamma}{\sigma} v_{r}\big)dr\Big\}\Big)_{+}+\frac{1}{2}\int_0^t\big(\frac{\gamma}{\sigma} v_s\big)^2\,ds\,\Big] \\
&= \inf_{v\in\cV_t }E_{}\Big[\,-\beta \sigma \Big(  \max_{0\leq s\leq t}\Big \{\frac{c- z}{\sigma} - W_{s}-  \int_{0}^{s}v_{r}dr\Big\}\Big)_{+}+\frac{1}{2}\int_0^t v^{2}_{s}\,ds\,\Big], 
\end{aligned} 
$$
where we used the fact that $v\in \cV_t$ if and only of $\frac{\gamma}{\sigma} v \in  \cV_t $ in the last equality. 
Define 
$$
F\Big(W_{s} +\int_{0}^{s}v_{r}dr\Big) =-\beta \sigma\Big(  \max_{0\leq s\leq t}\Big \{\frac{c- z}{\sigma} - W_{s}-  \int_{0}^{s}  v_{r} dr\Big\}\Big)_{+}.
$$
Clearly the functional $F(\cdot)$ is bounded from above by $0$, so we can use the Boue-Dupuis variational formula (Theorem 5.1 from \cite{Boue:1998aa}) to get that the trader's value function is given by
\begin{equation*} 
V(t,z)= \frac{1}{\beta} \log \Big(E\Big[\exp\Big(\Big(  \beta \sigma \max_{0\leq s\leq t} \Big\{\frac{c- z}{\sigma}- W_{s}\Big\}\Big)_{+}\Big)\Big] \Big),\qquad z\ge 0.
\end{equation*}
Since $ \big( \max_{0\leq s\leq t} \Big\{\frac{c- z}{\sigma}- W_{s}\Big\}\big)_{+}$ and $L^{(c- z)/\sigma}_{t}(W)$ have the same law under $P$, we get (\ref{V-00}). 
 
We write $S^{v,z}$ for the reflected semimartingale in \eqref{g-proc}, with $S_0=z > c$ and with a given   $v\in \cV_t $. It\^o's formula and an application of Proposition \ref{prop-sing-pde}
 give that for any $\eps\in(0,t)$ 
\begin{align}  
U (\eps,S^{v,z}_{t-\eps}) &= U (t,z) 
+\sigma\int_0^{t-\eps}\partial _z U (t-r,S^{v,z}_{r})\,dW_{r} +\gamma\int_0^{t-\eps}\partial _z U (t-r,c)\,dY^{v}_{r}\nonumber\\
&\quad+\int_0^{t-\eps}\Big(-\partial_t U(t-r,S^{v,z}_{r}) +\gamma v(r) \partial_zU (t-r,S^{v,z}_{r})+\frac{\sigma^2}2 \partial_{zz}U(t-r,S^{v,z}_{r})\Big)\,dr\nonumber\\
&= U (t,z) 
+\sigma\int_0^{t-\eps}\partial _z U (t-r,S^{v,z}_{r})\,dW_{r} -\gamma Y^{v}_t\nonumber\\
&\quad+\int_0^{t-\eps}\Big( \gamma v_r\partial_zU (t-r,S^{v,z}_r) -\frac{\gamma^2}{4\kappa} \big(\partial_zU(t-r,S^{v,z}_{r})\big)^{2}\Big)\,dr\nonumber\\
&\le   U (t,z) 
+\sigma\int_0^{t-\eps}\partial _z U (t-r,S^{v,z}_{r})\,dW_{r} -\gamma Y^{v}_t+\kappa\int_0^{t-\eps} v_r^2\,dr .\label{verification arg eq}
  \end{align}
 It follows from \eqref{first U derivatives} that $\partial_z U $ is bounded, and so $\sigma\int_0^{t-\eps}\partial _z  U (t-r,S^{v,z}_{r})\,dW_{r}$ is a true $P$-martingale. Taking expectations hence gives
 $$U (t,z)\ge E\big[ U (\eps,S^{v,z}_{t-\eps}) \big]+E\big[ \,\gamma Y^{v}_{t-\eps} \big]-E\Big[\kappa\int_0^{t-\eps} v_r^2\,dr\,\Big].
 $$
  Using that  $U (\eps,\cdot)\to0$ boundedly as $\eps\downarrow 0$ and monotone convergence for the two rightmost expectations yields
  \begin{equation}  \label{rf100}
\begin{aligned}
U (t,z)\ge E\Big[\,\gamma Y^{v}_t-\kappa\int_0^t v_r^2\,dr\,\Big].
\end{aligned}
\end{equation}
Taking the supremum over $v\in\cV_t $ shows the inequality \lq\lq$\ge$\rq\rq\ in  Theorem~\ref{theorem-trader}. 

Now we consider the reflecting SDE \eqref{skor-sde}. To this end, we note first that, on any domain $[\eps,T]\times[c,\infty]$ with $\eps\in(0,T)$, the function $\bar v(t,z)$ is bounded and satisfies a global Lipschitz condition in $z$.  Next, we define 
$$\Gamma_t f:=f(t)+\max_{0\le r\le t}\{c-f(r)\},\qquad f\in C[0,T].
$$
Since $\Gamma$ is Lipschitz continuous with respect to the supremum norm, there exists for any $\eps\in(0,T)$ a unique strong solution to the SDE
\begin{equation}\label{bar S SDE}
d\bar S_t=\sigma \,dW_t+\gamma\bar v(T-t,\Gamma_t \bar S)\,dt,\qquad \bar S_0=z,
\end{equation}
on the time interval $[0,T-\eps]$. Uniqueness implies that this solution extends to $[0,T)$, and since $\bar v$ is bounded and non-positive, the limit $\bar S_T:=\lim_{t\uparrow T}\bar S_t$ exists and is finite. That is, the SDE \eqref{bar S SDE} admits a unique strong solution on $[0,T]$. When letting $v^*_t:=\bar v(T-t,\Gamma_t\bar S)$, then $v^*\in\cV_T$. Moreover,
$$S_t:=\bar S_t+\big(\max_{0\le r\le t}\{c- \bar S_r\}\big)_+=\Gamma_t\bar S
$$
solves the reflecting SDE \eqref{skor-sde} for $R_t=\big(\max_{0\le r\le t}\{c-\bar S_r\}\big)_+$. Finally, with this choice of $v^*$, we have an equality in \eqref{verification arg eq} and hence in \eqref{rf100}. This concludes the proof.
\qed

\section*{Data availability statement}

The data used in this article is publicly available from the Bank of England and the Swiss National Bank. The authors will be happy to share the used data upon request.

\section*{Acknowledgments}
We are very grateful to Amarjit Budhiraja whose numerous useful comments enabled us to significantly improve this paper.
We are also thankful to anonymous referees for careful reading of the manuscript and for a number of useful comments and suggestions that significantly improved this paper.

\bibliographystyle{abbrv}

\end{document}